\documentclass[conference]{IEEEtran}
\usepackage{amsthm, amssymb, amsmath}
\usepackage{verbatim, float}
\usepackage{mathrsfs}
\usepackage{graphics}
\usepackage[usenames,dvipsnames]{pstricks}
\usepackage{epsfig}
\usepackage{balance}
\usepackage{pst-grad} 
\usepackage{pst-plot} 
\usepackage{cases}
\usepackage{algorithmic}
\usepackage{bm}
\usepackage[noadjust]{cite}
\usepackage[top=1.2cm, bottom=0.86in, left=0.6in, right=0.6in]{geometry}
\usepackage{url}
\usepackage{subcaption}
\usepackage{graphicx} 

\usepackage{hhline}

\usepackage{algorithm}
\floatname{algorithm}{Construction}

\usepackage{setspace,lipsum}
\usepackage{etoolbox}
\makeatletter
\patchcmd{\@makecaption}
  {\scshape}
  {}
  {}
  {}
\makeatother

\usepackage[singlelinecheck=false]{caption}
\usepackage{diagbox}

\theoremstyle{plain}
\newtheorem{theorem}{Theorem}
\newtheorem{lemma}{Lemma}

\newtheorem{corollary}{Corollary}
\newtheorem{proposition}{Proposition}

\theoremstyle{definition}
\newtheorem{definition}{Definition}
\newtheorem{example}{Example}
\newtheorem{remark}{Remark}

\let\oldabs\abs
\def\abs{\@ifstar{\oldabs}{\oldabs*}}

\let\oldnorm\norm
\def\norm{\@ifstar{\oldnorm}{\oldnorm*}}
\makeatother
\newcommand{\tr}{{\rm Tr}}
\newcommand{\rs}{{\rm RS}}
\newcommand{\grs}{{\rm GRS}}

\newcommand{\FF}{\mathbb{F}}
\newcommand{\BB}{\mathbb{B}}
\newcommand{\C}{\mathcal{C}}

\newcommand{\A}{\mathcal{A}}

\newcommand{\bx}{\pmb{x}}
\newcommand{\ba}{\pmb{a}}
\newcommand{\bb}{\pmb{b}}
\newcommand{\bc}{\pmb{c}}
\newcommand{\bes}{\beta^*}

\newcommand{\cj}{c_j}
\newcommand{\ft}{\mathbb{F}_2}
\newcommand{\ff}{\mathbb{F}_4}
\newcommand{\fq}{\mathbb{F}_q}

\newcommand{\fql}{\mathbb{F}_{q^\ell}}

\newcommand{\fa}{f(\alpha)}
\newcommand{\faj}{f(\alpha_j)}
\newcommand{\fb}{f(\beta)}
\newcommand{\qba}{\mathsf{Q}^{(\beta)}(\alpha)}
\newcommand{\qbha}{\mathsf{Q}^{(\beta)}_z(\alpha)}
\newcommand{\qboa}{\mathsf{Q}^{(\beta)}_1(\alpha)}
\newcommand{\qbaj}{\mathsf{Q}^{(\beta)}(\alpha_j)}
\newcommand{\qbpa}{\mathsf{Q}^{(\beta')}(\alpha)}
\newcommand{\rba}{\mathsf{R}^{(\beta)}(\alpha)}
\newcommand{\rboa}{\mathsf{R}^{(\beta)}_1(\alpha)}
\newcommand{\rbha}{\mathsf{R}^{(\beta)}_z(\alpha)}

\newcommand{\la}{\lambda_\alpha}
\newcommand{\al}{\alpha}
\newcommand{\be}{\beta}

\DeclareMathOperator{\im}{Im}
\DeclareMathOperator{\spn}{span_{\mathbb{B}}}

\DeclareMathOperator{\pr}{Pr}

\begin{document}

\title{Private Repair of a Single Erasure\\ in Reed-Solomon Codes}

\author{Stanislav Kruglik$^\dag$, Han Mao Kiah$^\dag$, Son Hoang Dau$^\ddag$, and Eitan Yaakobi$^{\dag*}$\\$^\dag$Nanyang Technological University, $^\ddag$RMIT University, $^*$Technion\\
 {\{stanislav.kruglik, hmkiah\}@ntu.edu.sg, sonhoang.dau@rmit.edu.au, yaakobi@cs.technion.ac.il}}

\maketitle

\begin{abstract}
We investigate the problem of \textit{privately} recovering a single erasure for Reed-Solomon codes with low communication bandwidths. For an $[n,k]_{\fql}$ code with $n-k\geq q^{m}+t-1$, we construct a repair scheme that allows a client to recover an arbitrary codeword symbol without leaking its index to any set of $t$ colluding helper nodes at a \textit{repair bandwidth} of $(n-1)(\ell-m)$ sub-symbols in $\FF_q$. When $t=1$, this reduces to the bandwidth of existing  repair schemes based on subspace polynomials. We prove the optimality of the proposed scheme when $n=q^\ell$ under a reasonable assumption about the schemes being used. Our private repair scheme can also be transformed into a private \textit{retrieval} scheme for data encoded by Reed-Solomon codes. 
\end{abstract}

\section{Introduction}
\label{sec:intro}

The problem of recovering a single erasure for Reed-Solomon codes~\cite{originalRSpaper}, currently widely used in distributed storage systems~\cite{DinhNguyenMohanSerdarLuongDau_ISIT2022}, has attracted considerable efforts from the research community in the past few years~\cite{Shanmugam2014,guruswami2016repairing,guruswami2017repairing,dau2017optimal,DauDuursmaKiahMilenkovicTwoErasures2017,YeBarg_ISIT2016,YeBarg_TIT2017,tamo2017optimal,TamoYeBarg2018,dau2018repairing,DuursmaDau2017,LiWangJafarkhani-Allerton-2017,li2019sub,LiWangJafarkhani_CommLett_2021,ChowdhuryVardy2017,ChowdhuryVardy_TIT_2019,DauViterbo-ITW-2018,DauDuursmaChu-ISIT-2018,LiDauWangJafarkhaniViterbo_ISIT2019,  ZhangZhang_ISCIT_2019,XuZhangWangZhang_2023,DinhBoztasDauViterbo_2023,DinhNguyenMohanSerdarLuongDau_ISIT2022,ConTamo_2022,ConShuttyTamoWootters_2023,BermanBuzagloDorShanyTamo_ISIT_2021,ShuttyWootters_2022,kiah2023explicit, ISIT2024}. In its typical setting, given a finite field $\FF$, $k$ data objects $\bx=(x_1,\ldots,x_k) \in \FF^k$ are first transformed into $n>k$ codeword symbols $\bc=(c_1,c_2,\ldots,c_n)\in \FF^n$ using a Reed-Solomon code, which then are stored at $n$ different storage nodes (servers). The key question is to find a repair scheme to recover a codeword symbol $c_{j^*}$ stored at Node $j^*$ by extracting as little data as possible from other $\cj$, $j\neq j^*$, stored at the helper nodes. The total amount of information \textcolor{black}{downloaded from helper node to complete the repair process}\footnote{\textcolor{black}{We note that focusing on download cost instead of overall communications can be motivated by the fact that each node can store several Reed-Solomon entities corresponding to the same evaluation point. In this case, we can perform repair through a single request, making the upload cost negligible compared to the download cost.}} is referred to as the \textit{repair bandwidth} in the literature, which should be minimized to reduce the communication cost. In this work, we are interested in a \textit{privacy} issue of the repair process and ask the following question: can the repair node \textit{hide} the index $j^*$ of the node being repaired from at most $t$ colluding curious helper nodes? A toy example of such a scheme with $t=1$ is given in Fig.~\ref{fig:toy_example}.

Private repair schemes for Reed-Solomon codes have a direct application in private \textit{retrieval} of data encoded by Reed-Solomon codes, described as follows. Suppose that a storage system stores $k$ data objects $\bx=(x_1,\ldots,x_k)$ across $n$ server nodes, which are encoded by a Reed-Solomon code. Thanks to the MDS property of the code, a client can always retrieve an object without leaking its identity to any group of nodes by retrieving all the data from an arbitrary set of $k$ nodes and recover the object. However, this naive approach requires a high communication cost: it is as costly as retrieving all $k$ objects. Using a private repair scheme with low repair bandwidth, the client downloads data from other nodes according to this scheme while also \textit{pretending} to request some data from the node storing the desired object.
For example, in Fig.~\ref{fig:toy_example}~a), the client can pretend to download $a_2$ from the node storing $(a_1,a_2)$ to make it unable to determine whether $\ba$ or $\bb$ is being retrieved. \textcolor{black}{Importantly, this notion differs from the coded private information retrieval (PIR) setup considered in~\cite{PIR1, PIR2, PIR3} and references therein in the way we store data objects on nodes. This results in significantly lower storage on each node in absolute terms due to the smaller size of each data object (as low as one field element). At the same time, encoding data objects by systematic Reed-Solomon code allows us to create a storage system where a data object without a privacy requirement can be retrieved by downloading the whole content of one node, while with a privacy requirement, we can employ a private repair scheme. Note that the introduced redundancy can also be employed for failed node repair.}

We \textit{first} modify the repair schemes based on subspace polynomials (see~\cite{guruswami2016repairing, guruswami2017repairing, dau2017optimal, dau2021repairing}) to accommodate the privacy requirement, achieving a bandwidth of $(n-1)(\ell-m)$ sub-symbols in $\fq$ when repairing a single erasure in an $[n,k]_{\fql}$ Reed-Solomon code, given that $n-k\geq q^m+t-1$. 
Essentially, the repair node selects a repair scheme \textit{randomly} to confuse the helper nodes.
When $t=1$, i.e., privacy against one curious node is required, this reduces to the previously known repair bandwidth~\cite{dau2017optimal, dau2021repairing} (privacy against one node is free!). 
We \textit{then} prove that, under a reasonable assumption on private repair schemes employed, the achieved bandwidth is optimal. 
Note that while in the minimal example depicted in Fig.~\ref{fig:toy_example}, the \textit{retrieval} scheme based on the constructed private repair scheme incurs the same bandwidth (4 bits) as the naive approach, in general, our constructed repair scheme leads to a lower \textit{retrieval} bandwidth than that of the naive one in several parameter regimes. 
 
\begin{figure}[tb]
\centering
\includegraphics[scale=0.8]{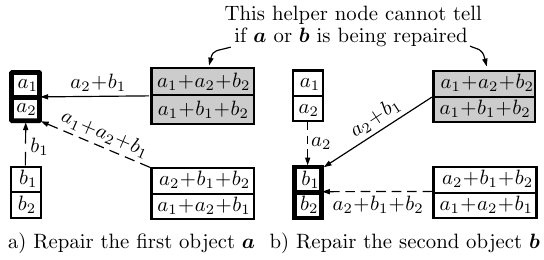}
\caption{A toy example of a private repair scheme operating in a distributed storage system that stores two objects $\ba=(a_1,a_2)\in \ff$ and $\bb=(b_1,b_2)\in \ff$ with two parities using a systematic $[4,2]_{\ff}$ Reed-Solomon code (see Example~\ref{ex:toy} for more details). The client, which plays the role of a repair node, requests $a_2+b_1=(a_1+a_2+b_2)+(a_1+b_1+b_2)$ from the shaded parity node, which cannot tell if the request is for repairing $\ba$ or $\bb$, as both cases are equally probable. A similar example can also be constructed for the bottom parity node. Such a repair scheme (based on trace polynomials) is private against one curious node.} 
\label{fig:toy_example}
\end{figure}

The paper is organized as follows. Necessary definitions and notations are introduced in Section~\ref{sec:preliminaries}.
Our construction of a private repair scheme is developed in Section~\ref{sec:scheme}. The lower bound on the private repair bandwidth is provided in Section~\ref{sec:bound}. We conclude the paper in Section~\ref{sec:conclusion}.

\section{Preliminaries}
\label{sec:preliminaries}

We first define basic notations and then describe the repair schemes based on subspace polynomials (\cite{guruswami2016repairing,guruswami2017repairing,dau2017optimal,dau2021repairing}).

\subsection{General Definitions and Notations}
\label{subsec:general_notations}

Let $[n]$ denote the set $\{1,2,\ldots,n\}$ and $\FF$ denote a finite field. 
A {\em linear $[n, k]_\FF$ code} $\C$ is a $k$-dimensional $\FF$-linear subspace of $\FF^n$. Each $\bc=(c_1,\ldots,c_n)\in\C$ is called a {\em codeword}, while $c_j$'s are called \textit{codeword symbols}. 
The orthogonal complement of $\mathcal{C}$ is called the {\em dual code} $\mathcal{C}^{\perp}$, which is an $[n,n-k]$ linear code over $\FF$. 
Hence, for every $\pmb{c}=(c_1,\ldots, c_n)\in\mathcal{C}$ and $\pmb{c}^{\perp}=(c_1^{\perp},\ldots, c_n^{\perp})\in\mathcal{C}^{\perp}$, it holds that $\sum_{j=1}^n c_j c_j^{\perp}=0$. 

\begin{definition}[\cite{originalRSpaper}]\label{def:rs}
Let $\FF[x]$ denote the ring of polynomials over $\FF$. The {\em Reed-Solomon code} $\rs(\mathcal{A}, k)$ of dimension $k$ with evaluation points $\mathcal{A}= \{\alpha_1, \dots , \alpha_n\}  \subseteq \FF$ 
is defined as:
\[
\rs(\mathcal{A}, k) \triangleq \left\{(f(\alpha_1), \ldots , f(\alpha_n)) \colon f \in \FF[x],\, \deg(f) < k\right\}. 
\]
On the other hand, the {\em generalized Reed-Solomon code} $\grs(\mathcal{A}, k, \pmb{\lambda})$ of dimension $k$ with evaluation points $\mathcal{W}$ and non-zero multipliers $\pmb{\lambda}=(\lambda_{\alpha_1},\ldots,\lambda_{\alpha_n})\in\FF^n$ is defined as:
    \begin{align*}
    &\grs(\mathcal{A}, k, \pmb{\lambda})\notag\\
   &\triangleq\left\{
	(\lambda_{\alpha_1}f(\alpha_1), \ldots , \lambda_{\alpha_n}f(\alpha_n))
	\colon f \in \FF[x],\, \text{deg}(f) < k\right\}.  
    \end{align*}
\end{definition}

Clearly, the generalized Reed-Solomon code $\grs(\mathcal{A}, k, \pmb{\lambda})$ with multiplier vector $\pmb{\lambda}=(1,\ldots,1)$ is a Reed-Solomon code $\textrm{RS}(\mathcal{A}, k)$ with the same set of evaluation points and dimension. Another important fact is that the dual of $\textrm{RS}(\mathcal{A}, k)$ is $\textrm{GRS}(\mathcal{A}, n-k, \pmb{\lambda})$, where the multiplier vector can be explicitly given as (see, for example \cite[Ch. 10, 12]{MacWilliams1977}) 
\begin{equation}\label{GRSM}
\lambda_{\alpha_i}^{-1}=\prod_{j\in [n]\setminus\{i\}}(\alpha_i-\alpha_j).    
\end{equation}

For convenience, when clear from context, $f(x)$ denotes a polynomial of degree at most $k-1$ corresponding to a codeword in $\rs(\mathcal{A},k)$, 
while $r(x)$ denotes a polynomial of degree at most $n-k-1$ corresponding to a dual codeword from $\grs(\mathcal{A}, n-k, \pmb{\lambda})$. 
The trace-repair framework for recovering one codeword symbol of Reed-Solomon codes introduced by Guruswami and Wootters in~\cite{guruswami2016repairing, guruswami2017repairing} relies on the existence of low-degree polynomials $r(x)$ and the equality
	$\sum_{j=1}^n \lambda_{\alpha_j} f(\alpha_j) r(\alpha_j) = 0$.

In what follows, we refer to polynomials $r(x)$ as {\em parity-check polynomials} for $\C$. Of interest, we employ {\em subspace polynomials}, introduced as a generalization of {\em trace polynomials} for Reed-Solomon code repair problem in~\cite{dau2017optimal, dau2021repairing}. Let us formally define them. Suppose that $\FF$ is a field extension of $\BB$ of degree $\ell > 1$. Then $\FF$ can also be treated as the vector space $\BB^\ell$. We refer to the elements of $\FF$ as {\em symbols} and the elements of $\BB$ as {\em sub-symbols}. Note that each symbol comprises of $\ell$ sub-symbols. Let us we consider the {\em trace function} $\tr \colon \FF\to \BB$ defined as 
	$\tr(x) = \sum_{j=0}^{\ell-1}x^{|\mathbb{B}|^j} \text{ for } x\in \FF$. 
It is well known that $\tr$ is a $\BB$-linear mapping from $\FF$ to $\BB$.

We can treat $\FF$ as a $\BB$-linear space of dimension $\ell$ and define a $\BB$-basis $\{u_1,\ldots,u_\ell\}$ for $\FF$. Moreover, from~\cite[Ch. 2]{Lidl}, there exists a {\em trace-dual basis} $\{\widetilde{u}_1,\ldots,\widetilde{u}_\ell\}$, which satisfies that $\tr(u_i\widetilde{u}_j)=1$ if $i=j$, and $\tr(u_i\widetilde{u}_j)=0$, otherwise, and
\begin{equation}
    \label{eq:trace_reconstruction}
x=\sum_{i=1}^\ell\tr(u_ix)\widetilde{u}_i 
\end{equation}    
for every $x \in \FF$, which says that every symbol $x\in \FF$ can be recovered from its $\ell$ (independent) traces $\tr(u_ix)$'s, where $\{u_1,\ldots,u_\ell\}$ forms a $\BB$-basis of $\FF$. 



\subsection{Repair Schemes Based on Subspace Polynomials}
\label{subsec:subspace}

\begin{definition}
Let $W$ be an $m$-dimensional $\BB$-subspace of $\FF$ ($0\hspace{-3pt} <\hspace{-3pt} m\hspace{-3pt} <\hspace{-3pt} \ell$). 
The subspace polynomial defined by $W$ is given by 
\begin{align*}
    L_W(x)\triangleq\prod_{\omega\in W}(x-\omega)=x\prod_{\omega\in W\setminus\{0\}}(x-\omega)=\sum_{j=0}^m l_jx^{|\mathbb{B}|^j},
\end{align*}
where $l_0 = - \prod_{w \in W\setminus\{0\}}w \neq 0$. Other coefficients $l_j$, $j> 0$, do not play any role in the repair scheme. 
\end{definition}

Note that $L_W$ is a $\mathbb{B}$-linear mapping from $\FF$ to itself with kernel $W$ (see, e.g.~\cite{Lidl}). The image of $L_W$ is a subspace of dimension $\ell-m$ over $\mathbb{B}$ and has a $\BB$-basis $\{\chi_1,\dots,\chi_{\ell-m}\}$. 

Suppose one wants to recover $\fb$, for some $\beta \in \A$. 
Pick a $\BB$-subspace $W$ of dimension $m$ in $\FF$, a $\BB$-basis $\{u_1,\ldots,u_\ell\}$ of $\FF$, and form the following polynomials
\begin{equation}\label{subspace-PCE}
    r_i(x) \triangleq \frac{L_W(u_i(x-\beta))}{x-\beta},\quad\text{for all }i=1,\ldots,\ell.
\end{equation}

It is clear that $\deg(r_i)=|\BB|^{m}-1$ and $r_i(\beta) = u_il_0$. If $|\BB|^{m}-1 < n - k$, then $r_i(x)$ are parity-check polynomials for the code and the following parity check equations hold
\begin{align}\label{eq:paritycheckSP}
    u_il_0\lambda_{\beta} \fb = -\sum_{\alpha\in\mathcal{A}\setminus\{\beta\}} r_i(\alpha)\lambda_{\alpha}\fa.
\end{align}
As $r_i(\alpha)\in \im(L_W)=\spn(\chi_1,\ldots,\chi_{\ell-m}\}$, we can write $r_i(\alpha) = \sum_{h\in [\ell-m]}\sigma_{i,\alpha,h}\chi_h$, for some $\sigma_{i,\alpha,h}\in \BB$, $h\in [\ell-m]$. Applying the trace function to both sides of \eqref{eq:paritycheckSP}, we obtain

\begin{align}
    \label{eq:repair_equations}
    &\tr(u_il_0\lambda_{\beta} \fb) = -\sum_{\alpha\in\mathcal{A}\setminus\{\beta\}} \tr(r_i(\alpha)\lambda_{\alpha}\fa)\notag\\
    &=-\sum_{\alpha\in\mathcal{A}\setminus\{\beta\}}\tr\bigg(\sum_{h\in [\ell-m]}\sigma_{i, \alpha, h}\chi_h\frac{\lambda_\alpha \fa}{\alpha-\beta}\bigg)\notag\\
    &=-\sum_{\alpha\in\mathcal{A}\setminus\{\beta\}} \sum_{h \in [\ell-m]}\sigma_{i,\alpha,h} \tr\left(\frac{\lambda_\alpha \fa\chi_h}{\alpha-\beta}\right).
\end{align}

From \eqref{eq:trace_reconstruction}, we know that $\fb$ can be reconstructed from $\tr(u_il_0\lambda_{\beta} \fb)$, $i\in [\ell]$, since $\{u_il_0\lambda_{\beta}\}_{i\in [\ell]}$ is a basis of~$\FF$.
Therefore, from \eqref{eq:repair_equations}, to repair $\fb$, it suffices to retrieve $\ell-m$ traces $\left\{\tr\left(\frac{\lambda_\alpha \fa\chi_h}{\alpha-\beta}\right)\right\}_{h\in [\ell-m]}$ of $\fa$ from the node storing $\fa$, $\alpha \in \A\setminus \{\beta\}$. 
Note that $\sigma_{i, \alpha, h}$'s are independent of $\fa$ and can be determined from $\alpha$, $\beta$, $W$, $u_i$, and $\{\chi_h\}_{h\in [\ell-m]}$.

\begin{proposition}[Subspace-polynomial repair scheme~\cite{guruswami2016repairing,guruswami2017repairing,dau2017optimal,dau2021repairing}]
    \label{pro:subspace_repair}
    Consider an $[n,k]_\FF$ $\rs(\mathcal{A}, k)$ with $\FF$ a degree-$\ell$ extension of $\BB$ satisfying $|\BB|^m \leq n-k$. To recover a codeword symbol $\fb$, where $f\in \FF[x]$, $\deg(f)<k$, and $\beta\in \A$, it suffices for the repair node to download $\ell-m$ traces $\left\{\tr\big(\frac{\lambda_\alpha \fa\chi_h}{\alpha-\beta}\big)\right\}_{h \in [\ell-m]}$ from the node storing $\fa$, for all $\alpha\in \A\setminus \{\beta\}$. Here, $W$ is an $m$-dimension $\BB$-subspace of $\FF$, $L_W(x)$ is its subspace polynomial, and $\{\chi_h\}_{h\in[\ell-m]}$ is a $\BB$-basis of $\im(L_W)$.
\end{proposition}
Following the repair in Proposition~\ref{pro:subspace_repair}, assuming that the subspace $W$, the polynomials $r_i$'s, and the basis $\{\chi_h\}_{h\in [\ell-m]}$ are publicly known, the client (repair node) who wants to recover $\fb$ simply sends $\beta$ to all the helper nodes. The node storing $\fa$, $\alpha\in \A\setminus \{\beta\}$, computes and sends $\tr\left(\frac{\lambda_\alpha \fa\chi_h}{\alpha-\beta}\right)$, $h\in [\ell-m]$, back to the client. 
We note that this approach immediately \textit{reveals} $\beta$ to all helper nodes. In the next section, we modify the scheme to hide $\beta$ from every set of at most $t$ colluding helper nodes ($t<k$).

\section{Private Repair for Reed-Solomon Codes}
\label{sec:scheme}

We first formally define private repair and then present the hidden-subspace scheme, which provides privacy against one helper node. 
Next, we establish a secret-sharing based repair scheme that is private against $t$-colluding helper nodes ($t\geq 1$).
In both approaches, a \textit{random} repair scheme must be used to confuse the helper nodes.   
Finally, we describe a transformation of such scheme to achieve private retrieval. 

\subsection{Definition of Private Repair}
\label{sec:motivation}

\begin{definition}
    \label{def:private_repair}
    Assume that the code $\rs(\A,k)$ and its parameters are publicly known, and that the node storing $\fa$ also knows~$\alpha$. Suppose the repair node, who wishes to repair $\fb$ for some $\beta\in \A$, sends a query vector $\qba = \big(\qboa,\ldots,\qbha\big)\in \FF^z$ to nodes storing $\fa$, $\alpha\in \A\setminus\{\beta\}$. Upon receiving $\qba$, the node storing $\fa$ responds with $\rba = (\rboa,\ldots,\rbha)$ $=\big(\tr(\qboa \fa),\ldots,\tr(\qbha\fa)\big)\in \BB^z$. The repair node uses $\{\rba\}_{\al\in \A\setminus\{\be\}}$ to recover $\fb$. The repair scheme is $t$-private if the following two conditions are satisfied.
    \begin{itemize}
        \item \textbf{(Correctness)} 
        \textcolor{black}{$\fb$ can be rebuilt from all the responses.}
        \item \textbf{(Privacy)} $\pr\left(\be = \bes \mid \{\qbaj\}_{j\in J}\right)=\frac{1}{n-|J|}$ for every $\bes\in \A\setminus \{\al_j\}_{j\in J}$, and every subset $\{\alpha_j\colon j \in J\}\subset \A\setminus \{\beta\}$, $|J|\leq t$. In other words, every set of $|J|\leq t$ requests does not leak any further information about $\beta$ apart from the knowledge that $\be \notin \{\al_j\colon j \in J\}$. 
    \end{itemize}  
\end{definition}

\subsection{A Hidden-Subspace Private Repair Scheme}
\label{subsec:hidden_subspace}

In a hidden-subspace repair scheme, the repair node randomly picks a subspace $W$ to generate the parity-check polynomials, thus effectively hiding the value to be repaired against every curious helper node. Such a scheme is $1$-private. \textcolor{black}{Before formulating the corresponding theorem, we need two technical lemmas on subspace polynomials. Note that Lemma~2 was proven in~\cite[Lemma~$3.7$]{lemma2}, however, the provided proof is different. For the convenience of the reader and completeness of the paper, we leave it here.}

\begin{lemma}
    \label{lem:different_W_different_imLW}
    If $W$ and $W'$ are two different $m$-dimensional $\fq$-subspaces of $\fql$ then $\im(L_W)\neq \im(L_{W'})$.
\end{lemma}
\begin{proof}
    We assume, for the sake of contradiction, that $W\neq W'$ but $U = \im(L_W) = \im(L_{W'})$. 
    Since $\ker(L_U)=U=\im(L_W)=\im(L_{W'})$, we deduce that $L_U\circ L_W \equiv 0$ and $L_U \circ L_{W'} \equiv 0$ on $\fql$, i.e., $L_U(L_W(\al)) = L_U(L_{W'}(\al)) = 0$ for every $\al\in \fql$. Therefore, $x^{q^s}-x = \prod_{\al\in \fql}(x-\al)$ divides both $L_U\circ L_W$ and $L_U \circ L_{W'}$. Since $\deg(L_U)=q^{s-m}$, $\deg(L_W)=\deg(L_{W'})=q^m$, we have $\deg(L_U\circ L_W)=\deg(L_U\circ L_{W'})=q^s$. Moreover, both are monic polynomials. Hence, $L_U\circ L_W \equiv L_U \circ L_{W'} = x^{q^s}-x$. This implies that $0 = L_U\circ L_W - L_U \circ L_{W'} = L_U \circ (L_W-L_{W'})$, which is impossible because $L_W - L_W' \neq 0$ due to the assumption that $W \neq W'$.
\end{proof} 

\begin{lemma}
    \label{lem:imLW}
    For every $(\ell-m)$-dimensional $\fq$-subspace $V$ of $\fql$, there exists a unique $m$-dimensional $\fq$-subspace $W$ of $\fql$ such that $V=\im(L_W)$.
\end{lemma}
\begin{proof}
    By Lemma~\ref{lem:different_W_different_imLW}, as $\im(L_W)\neq \im(L_{W'})$ if $W \neq W'$, the collection 
    \[
    S \triangleq \{\im(L_W) \colon W \text{ is an } m\text{-dimensional subspace of } \fql\}
    \]
    is a set of size 
    $
    \begin{bmatrix}
    \ell \\ m    
    \end{bmatrix}_q
    =
    \begin{bmatrix}
    \ell \\ \ell-m    
    \end{bmatrix}_q
    $. This means that $S$ is the set of all $\fq$-subspaces of dimension $\ell-m$ in $\fql$. Therefore, for every subspace $V$ of dimension $\ell-m$, there exists exactly one subspace $W$ of dimension $m$ so that $V=\im(L_W)$. 
\end{proof}

We now use the two lemmas above to prove the privacy of the repair scheme produced by Construction~I. Note that its correctness is guaranteed by Proposition~\ref{pro:subspace_repair}.

\begin{theorem}
    \label{thm:hidden_subspace}
    The scheme in Construction~\ref{construction:hidden_subspace} is a $1$-private repair scheme for one erasure of an $[n,k]_\FF$ Reed-Solomon code with $n-1$ helpers and a repair bandwidth of $(n-1)(\ell-m)$ sub-symbols in $\BB$, given that $n-k\geq |\BB|^m$.    
\end{theorem}

\begin{proof}[Proof of Theorem~\ref{thm:hidden_subspace}]
    The node storing $\fa$, $\al\in \A\setminus\{\be\}$, receives the query $\qba = \big( \eta_s\triangleq\frac{\la\chi_h}{\al-\be} \big)_{s \in [\ell-m]}$. Since $\{\chi_h\}_{s\in [\ell-m]}$ is linearly independent, so is $\{\eta_s\}_{s\in [\ell-m]}$. For privacy, it suffices to show that given the knowledge of $\{\eta_s\}_{s\in [\ell-m]}$, every $\be \in \A\setminus\{\al\}$ is equally probable. To this end, we aim to prove that for every $\be' \in \A\setminus\{\al\}$, there exists exactly one $m$-dimensional subspace $W'$ such that in the corresponding repair scheme obtained in Construction~I, $\qbpa = (\eta_s)_{s\in [\ell-m]} = \qba$. In other words, one needs to show that there exists an unique subspace $W'$ such that $\eta_s = \frac{\la \chi'_s}{\al-\be'}$ for all $s \in [\ell-m]$. The existence of such a subspace is guaranteed by Lemma~\ref{lem:imLW} with $V \triangleq \spn(\{\eta_s(\al-\be')/\la\}_{s\in [\ell-m]})$, which completes the proof.
\end{proof}

\begin{algorithm}[htb!]
Consider an $[n,k]_\FF$ $\rs(\mathcal{A}, k)$ with $\FF$ a degree-$\ell$ extension of $\BB$ satisfying $n-k \geq |\BB|^m$, and $\{u_i\}_{i\in [\ell]}$ a $\BB$-basis of $\FF$, both publicly known. The goal is to $1$-privately repair $\fb$, $\be\in\A$.
\begin{itemize}
    \item \textbf{(Preparation)} The repair node selects an $m$-dimensional subspace $W$ uniformly at random, computes its subspace polynomial $L_W(x)$, and constructs the corresponding parity-check polynomials $r_i(x)$'s following \eqref{subspace-PCE}. Moreover, it also selects an arbitrary basis $\{\chi_h\}_{h\in [\ell-m]}$ of $\im(L_W)$.
    \item \textbf{(Upload)} The repair node generates the query vector\footnotemark{} $\qba = \big(\frac{\la \chi_h}{\alpha-\beta}\big)_{h\in [\ell-m]}$, which are sent to the nodes storing $\fa$, for every $\al \in \A \setminus \{\be\}$. 
    \item \textbf{(Download)} Each node $\fa$, $\al \in \A\setminus \{\be\}$, computes and sends the repair node $\rba =\big(\tr\big(\frac{\lambda_\al\chi_h}{\al-\be}\fa\big)\big)_{h \in [\ell-m]}$.
    \item \textbf{(Recovery)} The repair node recovers $\ell$ traces of $\fb$ using the received traces via \eqref{eq:repair_equations} and then $\fb$ via \eqref{eq:trace_reconstruction}. 
\end{itemize}
\caption{(Hidden-Subspace $1$-Private Repair Scheme)}
\label{construction:hidden_subspace}
\end{algorithm}
\footnotetext{Compared to the original repair scheme (no privacy), the repair node sends $(\ell-m)\times$ more data. However, as the data is usually of large size, the increase in upload is negligible compared to the download and hence ignored.}

\begin{example}
    \label{ex:toy}
    Reusing~\cite[Ex.~1]{dau2018repairing}, let $\BB = \ft$, $\FF = \ff = \{0,1,\xi,\xi^2\}$, $\ell = 2$, $n = 4$, $k = 2$, and $m=1$, where $\xi^2+\xi+1 = 0$. Then $\{1,\xi\}$ is a
$\BB$-basis of $\FF$. Moreover, each element $\ba \in \ff$ can be represented
by a pair of bits $(a_1,a_2)$ where $\ba = a_1 + a_2\xi$. Suppose the data consists of two objects
$(\ba,\bb) \in \ff^2$. To devise a systematic RS code, we associate with $(\ba,\bb) \in \ff^2$ a polynomial $f(x)=f_{\ba,\bb}(x)\triangleq \ba + (\bb-\ba)x$.
Then $f(0) = a_1 + a_2\xi = \ba$, $f(1) = b_1 + b_2\xi = \bb$,
$f(\xi) = (a_1+a_2+b_2) + (a_1+b_1+b_2)\xi$, and
$f(\xi^2) = (a_2+b_1+b_2) + (a_1+a_2+b_1)\xi$.
The four codeword symbols $f(0)$, $f(1)$, $f(\xi)$, and $f(\xi^2)$ are stored at four different nodes as depicted in 
Fig.~\ref{fig:toy_example}.

To repair $\fb$, the repair node \textit{randomly} picks a $1$-dimensional $\ft$-subspace $W$ of $\ff$. Equivalently, $\im(L_W)=\mathsf{span}_{\ft}(\chi)$ for a \textit{random} element $\chi \in \ff^*$. The query vector $\qba=\big(\frac{\chi}{\al-\be}\big)$ (note that $\la=1$) is sent to the node $\fa$, to which any $\be \neq \al$ is equally probable due to the lack of knowledge of $\chi$. For example, when $\alpha = \xi$, as demonstrated in Fig.~\ref{fig:toy_example}, when receiving a query $\qba = (\xi)$ from the repair node, this node responds with $\rba = \tr\big(\xi f(\xi)\big)=a_2+b_1$. 
Note that $\xi = \frac{\xi^2}{\xi-0}=\frac{1}{\xi-1}=\frac{\xi}{\xi-\xi^2}$.
Thus, from its perspective, this request could be for repairing $f(0)$ with $\chi=\xi^2$, or $f(1)$ with $\chi=1$ (as shown in Fig.~\ref{fig:toy_example}), or $f(\xi^2)$ with $\chi=\xi$ (nodes $f(0)$, $f(1)$, and $f(\xi)$ send $a_1$, $b_2$, and $a_2+b_1$, respectively). Thus, it cannot tell which of these three nodes is being repaired.
\end{example}

The hidden-subspace approach does \textit{not} work for $t>1$. For example, with $\ell-m=1$ and $\im(L_W)=\spn(\{\chi\})$, two colluding helper nodes storing $f(\alpha_1)$ and $f(\alpha_2)$, who receive the queries $\frac{\lambda_{\alpha_1}\chi}{\alpha_1-\be}$ and $\frac{\lambda_{\alpha_2}\chi}{\alpha_2-\be}$, can determine $\beta$ immediately by calculating the ratio of the two queries, assuming that the knowledge of $\lambda_{\alpha}$, $\alpha\in \A$, and $\chi$ is public.

\subsection{A Secret-Sharing Based Private Repair Scheme}
\label{sec:trace-scheme}

To generate a random set of parity-check polynomials, apart from selecting a random subspace $W$ and creating parity-check polynomials $\{r_i(x)\}_{i\in [\ell]}$ based on $L_W(x)$ as in \eqref{subspace-PCE}, an alternative approach is to keep $W$ fixed (and publicly known) while multiplying the $r_i(x)$'s with a randomly generated polynomial $R(x)$ of a relevant degree. More specifically, let 
\begin{equation}\label{eq:paritychecknew}
r_i(x) \triangleq \frac{L_{W}(u_i(x-\beta))}{x-\beta} R(x),\quad\text{for }i=1,\ldots,\ell,
\end{equation}
where $R(x) \triangleq \sum_{s=0}^{t-1} R_s x^s \in\FF[x]$ and $\{R_s\}_{0\leq s \leq t-1}$ are $t$ independent uniformly random variables from $\FF$. 
Under the assumption that $n-k\geq |\BB|^m+t-1$, such $r_i(x)$'s are parity-check polynomials for the code. 
Intuitively, as $R(x)$ provides a $t$-degree of randomness, the resulting repair scheme is private against every set of at most $t$ colluding helpers.
The scheme is formally described in Construction~II. Its correctness and privacy are guaranteed by Theorem~\ref{main-theorem}. 
As evidenced from the statement of the theorem, the price to pay for privacy is a higher level of redundancy, i.e., a larger value of $n-k$.

\begin{theorem}\label{main-theorem}
The scheme in Construction~\ref{construction:secret_sharing} is a $t$-private repair scheme for one erasure of an $[n,k]_\FF$ Reed-Solomon code with $n-1$ helpers and a bandwidth of $(n-1)(\ell-m)$ sub-symbols in $\BB$, given that $n-k\geq|\BB|^{m}+t-1$.
\end{theorem}

\begin{proof}
First, it is obvious that the repair node downloads $\ell-m$ sub-symbols in $\BB$ from each helper node and therefore the bandwidth is $(n-1)(\ell-m)$. We now proceed to demonstrate the correctness and privacy of the scheme.

\textbf{Correctness.} We aim to show that the scheme correctly recovers $\fb$. To this end, we prove that the equation~\eqref{eq:trace-download} in the \textbf{Recovery} step of Construction~\ref{construction:secret_sharing} holds for all $i\in[\ell]$. In other words, \eqref{eq:trace-download} means that the repair node can reconstruct all the $\ell$ traces necessary for recovering $\fb$ using what it downloads from the helper nodes $(\tau_{j,h})$ and what it already knows $(\sigma_{i,j,h})$. 
Note that since $(\fa)_{\alpha\in\mathcal{A}}$ belongs to $\rs(\mathcal{A},k)$, we consider its dual code $\grs(\mathcal{A},k,\mathbf{\lambda})$ where $\mathbf{\lambda}$ is defined by~\eqref{GRSM}. 

It is clear that $r_i(x)$'s defined in~\eqref{eq:paritychecknew} are polynomials of degrees at most $|\BB|^{m}+t-2$ and $r_i(\beta)=u_il_0R(\beta)$, $i\in [\ell]$. As $|\BB|^{m}+t-2 < n - k$, $r_i(x)$'s correspond to dual codewords in $\grs(\mathcal{A},k,\mathbf{\lambda})$. Thus, the following parity check equations hold
\begin{align}\label{eq:paritycheckGW-new}
    u_il_0\lambda_{\beta}R(\beta)\fb =- \sum_{j=1}^{n-1} r_i(\alpha_j)\lambda_{\alpha_j}f(\alpha_j),\quad i \in [\ell]. 
\end{align}
Applying trace function to both sides of \eqref{eq:paritycheckGW-new}, using the linearity of trace and noting that $L_W(u_i(\al_j-\be))=\sum_{h=1}^{\ell-m}\sigma_{i,j,h}\chi_h$ (\textbf{Preparation}), $\kappa_j \triangleq \frac{R(\al_j)}{\al_j-\be}$, $j\in [n-1]$ (\textbf{Preparation}), and $\tau_{j,h}\triangleq \tr\left(\kappa_j\chi_{h}\lambda_{\alpha_j} f(\alpha_j)\right)$, $h\in[\ell-m]$ (\textbf{Download}), we obtain
\begin{align*}
\label{eq:trace-repair-new}
    &\tr(u_il_0\lambda_{\beta}R(\beta)\fb)\notag 
    =-\sum_{j=1}^{n-1} \tr(r_i(\alpha_j)\lambda_{\alpha_j}f(\alpha_j))\notag\\
    &= -\sum_{j=1}^{n-1} \tr\left(\left(\sum_{h=1}^{\ell-m}\sigma_{i,j,h}\chi_h\right) \frac{R(\alpha_j)}{\alpha_j-\beta}\lambda_{\alpha_j}f(\alpha_j)\right)\notag\\
    &= -\sum_{j=1}^{n-1}\sum_{h=1}^{\ell-m} \sigma_{i,j,h}\tr\left(\chi_h\kappa_j\lambda_{\alpha_j}f(\alpha_j)\right)
    = -\sum_{j=1}^{n-1}\sum_{h=1}^{\ell-m} \sigma_{i,j,h}\tau_{j,h}, 
\end{align*}
which shows that \eqref{eq:trace-download} is correct, as desired.

\textbf{Privacy.} It suffices to show that given the set of queries received by an arbitrary coalition of $t$ nodes storing $\{\faj \colon j \in J\}\subset \A\setminus \{\be\}$, $|J|=t$, no extra information about $\beta$ can be obtained apart from the knowledge that $\be \notin \{\faj \colon j \in J\}$. 
To be more specific, following Definition~\ref{def:private_repair}, we will demonstrate that given the knowledge of any arbitrary set of $t$ queries, every $\be$ not corresponding to these queries is equally probable. 

As $\kappa_j=\frac{R(\alpha_j)}{\alpha_j-\beta}=\frac{\sum_{s=0}^{t-1}R_s\al_j^s}{\alpha_j-\beta}$,
we have $\kappa_j\beta+\sum_{s=0}^{t-1}R_s\al_j^s=\kappa_j\alpha_j$. As $\beta$ and $\{R_s\}_{0\leq s \leq t-1}$ are unknown to the helper nodes, this can be treated as a linear equation with $t+1$ unknowns. 
Without loss of generality, we assume that $J = [t]$ and the coalition of these $t$ helpers received $\kappa_1,\ldots,\kappa_{t}$ as queries and form the following system of linear equations.

\begin{equation*}
\begin{bmatrix}
\kappa_1 & 1 & \ldots & \alpha_1^{t-1} \\
\vdots & \vdots & \ddots & \vdots \\
\kappa_{t} & 1 & \ldots & \alpha_{t}^{t-1} 
\end{bmatrix} \begin{bmatrix}
\beta \\
R_0 \\
\vdots \\
R_{t-1} 
\end{bmatrix}= \begin{bmatrix}
\kappa_1\alpha_1 \\
\kappa_2\alpha_2 \\
\vdots \\
\kappa_{t}\alpha_{t} 
\end{bmatrix},
\end{equation*}

\noindent which is equivalent to 
\begin{equation}\label{system-private}
\begin{bmatrix}
1 & \ldots & \alpha_1^{t-1} \\
\vdots & \ddots & \vdots \\
1 & \ldots & \alpha_{t}^{t-1} 
\end{bmatrix} \begin{bmatrix}
R_0 \\
\vdots \\
R_{t-1} 
\end{bmatrix}= \begin{bmatrix}
\kappa_1(\alpha_1-\beta) \\
\kappa_2(\alpha_2-\beta) \\
\vdots \\
\kappa_{t}(\alpha_{t}-\beta)
\end{bmatrix}.    
\end{equation}
Given the knowledge of $\kappa_1,\ldots,\kappa_{t}$ and $\alpha_1,\ldots,\alpha_{t}$, for every value of $\beta \notin\{\al_1,\ldots,\al_t\}$, one can solve the system~\eqref{system-private} uniquely. 
This means that every $\beta \notin\{\al_1,\ldots,\al_t\}$ is equally probable from the perspective of this coalition of $t$ helper nodes, which implies the privacy according to Definition~\ref{def:private_repair}.
\end{proof}

\begin{algorithm}[htb]
Consider an $[n,k]_\FF$ $\rs(\mathcal{A}, k)$ with $\FF$ a degree-$\ell$ extension of $\BB$,  $n-k \geq |\BB|^m + t - 1$ ($t\geq 1$), an $m$-dimensional subspace $W$, its subspace polynomial $L_W(x)$, a basis $\{\chi_h\}_{h\in [\ell-m]}$ of $\im(L_W)$, and a $\BB$-basis $\{u_i\}_{i\in[\ell]}$ of $\FF$ and its trace-dual basis $\{\widetilde{u}_i\}_{i\in [\ell]}$, all of which are publicly known. The goal is to $t$-privately repair $\fb$, where $\be = \al_n \in \A$ (relabelling if necessary).
\begin{itemize}
    \item \textbf{(Preparation)} The repair node first computes $\sigma_{i,j,h}\in \BB$, for $i\in [\ell]$, $j\in [n-1]$, and $h\in [\ell-m]$, so that $L_W(u_i(\al_j-\be))=\sum_{h=1}^{\ell-m}\sigma_{i,j,h}\chi_h$. It then chooses $t$ independent uniformly random variables $\{R_s\}_{0\leq s \leq t-1}$ from $\FF$ and construct a random polynomial $R(x) \triangleq \sum_{s=0}^{t-1} R_s x^s \in\FF[x]$. It also computes $R(\be)$ and $\kappa_j \triangleq \frac{R(\al_j)}{\al_j-\be}$, $j\in [n-1]$.
    \item \textbf{(Upload)} The repair node sends $\kappa_j$ as a query to the node storing $\faj$, for every $j \in [n-1]$. 
    \item \textbf{(Download)} Each node $\faj$, $j\in [n-1]$, computes and sends the repair node $\tau_{j,h}\triangleq \tr\left(\kappa_j\chi_{h}\lambda_{\alpha_j} f(\alpha_j)\right)\in\BB$, $h\in[\ell-m]$, where $\lambda_{\alpha_j}$ are the code's multipliers as in~\eqref{GRSM}.
    \item \textbf{(Recovery)} The repair node first computes the $\ell$ traces of $l_0\lambda_\be R(\be)\fb$ with respect to $\{u_i\}_{i\in [\ell]}$ as follows.
    \begin{equation}\label{eq:trace-download}
    \tr(u_il_0\lambda_{\beta}R(\beta)\fb)=-\sum_{j=1}^{n-1}\sum_{h=1}^{\ell-m}\sigma_{i,j,h}\tau_{j,h},\ i \in [\ell].    
    \end{equation}
    It then recovers $l_0\lambda_{\beta}R(\beta)\fb$ using \eqref{eq:trace_reconstruction}, which gives $\fb$, i.e., $\fb=\frac{\sum_{i=1}^\ell\tr(u_il_0\lambda_{\beta}R(\beta)\fb)\widetilde{u}_i}{l_0\lambda_{\beta}R(\beta)}$.
\end{itemize}
\caption{(Secret-Sharing $t$-Private Repair Scheme)}
\label{construction:secret_sharing}
\end{algorithm}

\begin{example}\label{ex:toy-2}
Let $\BB=\FF_2$, $\FF=\FF_8=\{0,1,\xi,\xi^2,1+\xi^2,1+\xi+\xi^2,1+\xi,\xi+\xi^2\}$, $\ell=3$, $n=8$, $k=5$, $m=1$ and $t=2$, where $\xi^3+\xi^2+1=0$. Then $\{1,\xi,\xi^2\}$ is a $\BB$-basis of $\FF$. Moreover, each element $\ba\in\FF_8$ can be represented as three bits $(a_1, a_2, a_3)$, where $\ba=a_1+a_2\xi+a_3\xi^2$. Let us choose a $1$-dimensional $\FF_2$-subspace $W=\{0, 1\}$ of $\FF_8$ with $L_W(x)=x(x-1)$. Clearly, $\im(L_W)$ has a basis $\{\chi_1, \chi_2\}$, where $\chi_1=1+\xi$ and $\chi_2=\xi+\xi^2$. Suppose the data corresponds to $f(x)=x^5+1$ and the codeword symbols $f(0),f(1),\ldots,f(\xi+\xi^2)$ are stored. 

To repair $\fb=f(\xi + \xi^2)$, the repair node chooses two independent random variables $R_0$ and $R_1$ from $\FF$. The query $\kappa_j = \frac{R(\alpha_j)}{\alpha_j-\beta}=\frac{R_0 + R_1\alpha_j}{\alpha_j - \beta}$ is sent to the node $f(\alpha_j)$. Then for any two nodes, every $\beta$ other than the $\al_j$'s of these nodes is equally probable due to the lack of knowledge of $R_0$ and $R_1$.
For example, nodes $f(0)$ and $f(1)$ receive queries $\kappa_1=\xi^2+\xi$ and $\kappa_2=1$ from the repair node.
However, from their perspectives, $\kappa_1$ and $\kappa_2$ can be queries for repairing $f(\xi)$ with $R_0=1$ and $R_1=\xi$, or for repairing $f(\xi^2)$ with $R_0=\xi$ and $R_1=\xi^2+\xi+1$, etc. 
Thus, $f(0)$ and $f(1)$ cannot tell which of the other six nodes is being repaired. Note that the proposed construction ensures $2$-private repair with a bandwidth of $14$ bits, while the naive approach requires $15$ bits.
\end{example}

\subsection{Private Retrieval of a Reed-Solomon-Coded Symbol}
\label{subsec:private_retrieval}A private \textit{retrieval scheme} for a Reed-Solomon codeword symbol is quite similar to a private repair scheme, except that all codeword symbols are available and the retrieval node should send enquiries to all $n$ nodes. \textcolor{black}{This problem clearly has similarities with the coded PIR setup~\cite{PIR1, PIR2, PIR3}. However, the way the data objects are encoded in the latter resulted in large storage on each node in absolute terms. In contrast, in our schemes, each node stores just a single field element.}

\begin{definition}
    \label{def:private_retrieval}
    Assume that the code $\rs(\A,k)$ and its parameters are publicly known, and that the node storing $\fa$ also knows~$\alpha$. Suppose the retrieval node, who wishes to retrieve $\fb$ for some $\beta\in \A$, sends a query vector $\qba = \big(\qboa,\ldots,\qbha\big)\in \FF^z$ to nodes storing $\fa$, $\alpha\in \A$. Upon receiving $\qba$, the node storing $\fa$ responds with $\rba = (\rboa,\ldots,\rbha)$ $=\big(\tr(\qboa \fa),\ldots,\tr(\qbha\fa)\big)\in \BB^z$. The repair node uses $\{\rba\}_{\al\in \A}$ to recover $\fb$. The retrieval scheme is $t$-private if the following two conditions are satisfied.
    \begin{itemize}
        \item \textbf{(Correctness)} 
        \textcolor{black}{$\fb$ can be rebuilt from all the responses.}
        \item \textbf{(Privacy)} $\pr\left(\be = \bes \mid \{\qbaj\}_{j\in J}\right)=\frac{1}{n}$ for every $\bes\in \A$, and every subset $\{\alpha_j\colon j \in J\}\subset \A$, $|J|\leq t$. In other words, every set of $|J|\leq t$ requests does not leak any information about $\beta$. 
    \end{itemize}  
\end{definition}

\begin{corollary}
\label{cr:PIR}
The scheme in Construction~\ref{construction:secret_sharing} can be transformed to a $t$-private repair scheme for one symbol of Reed-Solomon code $\rs(\mathcal{A},k)$ from all codesymbols with bandwidth $n(\ell-m)$ sub-symbols if  $n\geq|\BB|^{m}+k+t-1$.
\end{corollary}
\begin{proof} 
All steps of the scheme in Construction~\ref{construction:secret_sharing} remain the same, except that we also upload $\kappa_0=R(\beta)$ to node $\beta$ and download $\tau_{0,1}=\tr(\kappa_{0}\chi_1\lambda_{\beta}\fb),\ldots,\tau_{0,\ell-m}=\tr(\kappa_{0}\chi_{\ell-m}\lambda_{\beta}\fb)$ from it. Let us formally show that any coalition of up to $t$ nodes cannot gain any extra information about $\beta$. The main idea remains the same as before: we show that for any possible value of $\beta$, an adversary with access to $t$ nodes and the corresponding $\kappa$ can form a system to determine random coefficients from $R(x)$, and this system will have a single and unique solution. Without loss of generality, assume that the adversary obtains $\kappa_1, \ldots, \kappa_t$ corresponding to evaluation points $\alpha_1, \ldots, \alpha_t$. We need to distinguish two cases:

\begin{itemize}
    \item Adversary assume that $\beta\notin\{\alpha_1,\ldots,\alpha_t\}$.  In this case, according to the proof of Theorem~\ref{main-theorem}, for any value of $\beta$, the adversary can find a single and unique solution to the system that determines random coefficients from $R(x)$.
    \item Adversary assume that $\beta\in\{\alpha_1,\ldots,\alpha_t\}$. In this case, without loss of generality let him/her assume that $\beta=\alpha_1$. As a result, $\kappa_1=R(\alpha_1)$, and the system to determine coefficients of $R(x)$ takes the following form:
    \begin{equation*}
    \begin{bmatrix}
1 & \ldots & \alpha_1^{t-1} \\
\vdots & \ddots & \vdots \\
1 & \ldots & \alpha_{t}^{t-1} 
\end{bmatrix}\begin{bmatrix}
R_1 \\
\vdots \\
R_t 
\end{bmatrix}= \begin{bmatrix}
\kappa_1 \\
\kappa_2(\alpha_2-\alpha_1) \\
\vdots \\
\kappa_{t}(\alpha_{t}-\alpha_t)
\end{bmatrix}.    
\end{equation*}
Clearly, it has a single and unique solution. The same holds for any other possible values of $\beta$.
\end{itemize}
As a result, all values of $\beta$ are equiprobable, which concludes the proof.
\end{proof}

\section{A Lower Bound on Repair Bandwidth}
\label{sec:bound}

To obtain a concrete lower bound on the repair bandwidth of a private repair scheme, we will limit ourselves to a family of linear random schemes called \textit{canonical random repair schemes}, which includes the scheme from Construction~\ref{construction:secret_sharing} (but not Construction~\ref{construction:hidden_subspace}). This family has a compact representation and requires only one $\FF$-symbol to be sent to each helper node.

\begin{definition}
    \label{def:CRRS}
    Assume that the code $\rs(\A,k)$ and its parameters are publicly known, and that the node storing $\fa$ also knows~$\alpha$. A \textit{canonical random (linear) repair scheme} (CRRS) for repairing $\fb$ corresponds to a set of parity-check polynomials $\{r_i(x) \triangleq P_i(x)Q(x)R(x)\}_{i\in [\ell]}$, where $\{P_i(x)\}_{i\in [\ell]}\subset \FF[x]$ are polynomials depending on both $\beta$ and some publicly known information, $Q(x)$ is a (private) function involving $\beta$ in its formation, and $R(x)\in \FF[x]$ is a (private) random polynomial. Moreover, the repair node sends a single query $\qba = Q(\alpha)R(\alpha)$ to the nodes storing $\fa$, $\alpha\in \A\setminus\{\beta\}$. The node \textcolor{black}{storing} $\fa$ uses $\qba$\textcolor{black}{, $\fa$} and public information about $\{P_i(x)\}_{i\in [\ell]}$ to compute $\rba$ \textcolor{black}{and send it back to the repair node}. Finally, the repair node uses $\{\rba\}_{\al\in \A\setminus\{\be\}}$ to recover \textcolor{black}{the set of parity check polynomials $\{r_i(x)\}_{i\in [\ell]}$ and, as a result, determine} $\fb$.
\end{definition}

\begin{remark}
    The scheme in Construction~\ref{construction:secret_sharing} is a CRRS with $P_i(x)=L_W(u_i(x-\beta))$, $i\in[\ell]$, and $Q(x)=1/(x-\beta)$. The public information about $P_i$ is $L_W$ and $u_i$.
\end{remark}

Following the initial paper on repair problem of Reed-Solomon code~\cite{guruswami2016repairing}, 
we can represent any linear repair scheme for $\textrm{RS}(\mathcal{A}, k)$  code $\C$ by a 
matrix whose columns belong to the dual code $\C^{\perp}= \textrm{GRS}(\mathcal{A}, k, \mathbf{\lambda})$ and correspond to parity-check equations used in the repair. In what follows let us focus on linear repair schemes and give a formal definition of an evaluation matrix. 

\begin{definition}
An $n\times \ell$ matrix $\mathbb{M}^{(j^*)}$ whose entries belong to $\FF$ is called an evaluation matrix for codeword symbol $j^*$ ($j^*\in[n]$) of a code $\C$ with bandwidth $b$ if the following conditions hold:
\begin{itemize}
    \item Columns of $\mathbb{M}^{(j^*)}$ are codewords of the dual code $\C^{\perp}$,
    \item $\textrm{rank}_{\BB}\mathbb{M}^{(j^*)}[j^*;:]=\ell$, where $\mathbb{M}^{(j^*)}[j^*;:]$ denotes the set of elements in the row $j^*$ of $\mathbb{M}^{(j^*)}$,
    \item $\sum_{j\in[n]\setminus\{j^*\}}\textrm{rank}_{\BB}\mathbb{M}^{(j^*)}[j;:]=b$.
\end{itemize}
\end{definition}

Clearly, for each $i\in[\ell]$ we have that
\begin{equation}\label{PCE_mat}
 \sum_{j=1}^nf(\alpha_j)\mathbb{M}^{(j^*)}[j;i]=0,  
\end{equation}
where $\mathbb{M}^{(j^*)}[j;i]=\lambda_jr_i(\alpha_j)$ and $r_i(x)\in\FF[x]$ is parity-check polynomial. 

As we are considering a canonical random linear repair scheme, the repair node uses $r_i(x)=P_i(x)Q(x)R(x)$ and sends $R(\alpha_j)Q(\alpha_j)$ to node $\fa$. 
In this case, $t$-private repair means that for any $t$ helper nodes values of $R(x)Q(x)$ given to them do not reveal $ Q(x)$ and, as a result, the repaired index. Let us show that to ensure $t$-privacy $R(x)$ must have a degree at least $t-1$. First, let us consider the case when degree of $R(x)$ is $t-1$ or, in other words, $R(x)=R_0+R_1x+\cdots+R_{t-1}x^{t-1}$. Without loss of generality let us assume that adversary knows $\kappa_1=R(\alpha_1) Q(\alpha_1),\ldots,\kappa_t=R(\alpha_t) Q(\alpha_t)$. Hence, he/she can form the following system of equations:

\begin{equation}\label{system}
 \left\{
\begin{array}{lr}
R_0+R_1\alpha_1+\cdots+R_{t-1}\alpha_1^{t-1}=\frac{\kappa_1}{ Q(\alpha_1)}\\
\;\;\;\;\;\;\;\;\;\;\;\;\;\;\;\;\;\;\;\;\;\;\;\vdots\\
R_0+R_1\alpha_t+\cdots+R_{t-1}\alpha_t^{t-1}=\frac{\kappa_t}{ Q(\alpha_t)}
\end{array}.
\right.
\end{equation}

Clearly, for any repaired index and, as a result, $ Q(x)$ and $ Q(\alpha_1),\ldots, Q(\alpha_t)$, adversary can find a unique values of $R_0,\ldots,R_{t-1}$ such that  $\kappa_1=R(\alpha_1)Q(\alpha_1),\ldots,\kappa_t=R(\alpha_t)Q(\alpha_t)$ and, hence, any values of repaired index are equiprobable. However, when the degree of polynomial $R(x)$
 is $t-2$ the system~\eqref{system} is transformed to
\begin{equation}\label{system2}
 \left\{
\begin{array}{lr}
R_0+R_1\alpha_1+\cdots+R_{t-2}\alpha_1^{t-2}=\frac{\kappa_1}{ Q(\alpha_1)}\\
\;\;\;\;\;\;\;\;\;\;\;\;\;\;\;\;\;\;\;\;\;\;\;\vdots\\
R_0+R_1\alpha_{t-1}+\cdots+R_{t-2}\alpha_{t-1}^{t-2}=\frac{\kappa_{t-1}}{ Q(\alpha_{t-1})}\\
R_0+R_1\alpha_t+\cdots+R_{t-2}\alpha_t^{t-1}=\frac{\kappa_t}{ Q(\alpha_t)}
\end{array}.
\right.
\end{equation}

It can be easily seen that this system is overdetermined meaning that adversary can obtain a unique $R_0,\ldots,R_{t-2}$ from the first $t-1$ equations of~\eqref{system2} and use the remaining equation to check either his/her choice of $ Q(x)$ is correct or not meaning that different $ Q(x)$ are not equiprobable and, as a result, values of repaired index is not equiprobable. The same logic works for lower values of degree of polynomial $R(x)$. As a result, in case of $t$-private repair, $R(x)$ must have degree $t-1$. This means that equation~\eqref{PCE_mat} can be written as
\begin{equation}
\sum_{j=1}^n\lambda_{\alpha_j}f(\alpha_j) R(\alpha_j) P_i(\alpha_j) Q(\alpha_j)=0.
\end{equation}
Therefore, we can form a matrix $\tilde{\mathbb{M}}^{(j^*)}$ as $\tilde{\mathbb{M}}^{(j^*)}[j,h]=\lambda_{\alpha_j} P_i(\alpha_j) Q(\alpha_j)$. Clearly, it is a repair matrix with bandwidth $b$ for Reed-Solomon code with polynomial $f(x) R(x)$ since columns of $\tilde{\mathbb{M}}^{(j^*)}$ are codewords of dual code and for fixed $j$ we have that
\begin{align*}
 &\textrm{rank}_{\BB}[\{R(\alpha_j) P_1(\alpha_j) Q(\alpha_j),\ldots,R(\alpha_j) P_i(\alpha_j) Q(\alpha_j)\}]\notag\\
 &=\textrm{rank}_{\BB}[\{P_1(\alpha_j) Q(\alpha_j),\ldots,P_i(\alpha_j)Q(\alpha_j\}]   
\end{align*}
and therefore, rank conditions for matrix $\tilde{\mathbb{M}}^{(j^*)}$ coincides with rank conditions for matrix $\mathbb{M}^{(j^*)}$. Consequently, the problem of $t$-private repair of $f(\alpha_i)$, where $\textrm{deg}(f(X))<k$, is equivalent to the problem of non-private repair of  $f(\alpha_i) R(\alpha_i)$,  where $\textrm{deg}(f(X) R(X))<k+t-1$. As a result, we can get the following estimation on bandwidth in symbols of $\BB$.
\begin{equation}\label{bound:fractional}
	b\ge b_{\min}=(n-1)\log_{|\BB|}\frac{|\FF|(n-1)}{(|\FF|-1)(n-k-t)+ n-1}\,.
\end{equation}
However, quantity~\eqref{bound:fractional} is obtained under assumption that each node contributes equally and fractional values of individual node contribution are possible. However, this bound can be improved by solving integer program problem as it was done in~\cite{DM} for single node repair without privacy guarantees. As a result, we can formulate the following theorem.

\begin{theorem}\label{thm:lower_bound}
Let $\C$ be an $[n,k]_{\FF}$ Reed-Solomon code. 
Then every $t$-private canonical random linear repair scheme for $\C$ over a subfield $\BB$ of $\FF$ must have bandwidth at least 
\begin{equation}\label{eq:bmin}
	b\hspace{-2pt}\geq\hspace{-2pt} b_{\min}\hspace{-2pt} =\hspace{-2pt} n_0\hspace{-2pt} \left\lfloor \log_{|\BB|} \frac{n-1}{L}\right\rfloor \hspace{-2pt}+\hspace{-2pt} (n-1-n_0)\hspace{-2pt}  \left\lceil \log_{|\BB|} \frac{n-1}{L} \right\rceil
\end{equation}
sub-symbols in $\BB$, where $L\hspace{-2pt}=\hspace{-2pt}\frac{1}{|\FF|}\hspace{-2pt}\Big((|\FF|\hspace{-2pt}-\hspace{-2pt}1)(n\hspace{-2pt}-\hspace{-2pt}k\hspace{-2pt}-\hspace{-2pt}t)\hspace{-2pt}+\hspace{-2pt} n\hspace{-2pt}-\hspace{-2pt}1\Big)$, and 
\begin{align}
	n_0 & = \left\lfloor \frac{L-(n-1)|\BB|^{-\left\lceil \log_{|\BB|} \frac{n-1}{L} \right\rceil}}{|\BB|^{-\left\lfloor \log_{|\BB|} \frac{n-1}{L}\right\rfloor} - |\BB|^{-\left\lceil \log_{|\BB|} \frac{n-1}{L}\right\rceil}}\right\rfloor.
\end{align}
\end{theorem}

\begin{remark}
In our case, we assume that repair of $R(\beta) \fb$ is the same as the repair of $\fb$ because we know the polynomial $R(x)$ and can easily compute $R(\beta)$. However, when this assumption does not hold, the last part of our derivations does not apply, and the bound may take a different shape.
\end{remark}

\begin{corollary}
When $n=|\FF|=q^\ell$ and $n=q^m+k+t-1$, every $t$-private canonical random linear repair scheme for $[n,k]_\FF$ Reed-Solomon code over the subfield $\BB=\fq$ must have bandwidth at least $(n-1)(\ell-m)$ sub-symbols in $\BB$.
\end{corollary}

\section{Numerical results}
\label{sec:numerical-results}
Let us consider the private repair scheme derived in this paper and the corresponding lower bound. Let us fix $\FF=\textrm{GF}(2^8)$ since this field is popular in practice, and $\BB=\textrm{GF}(2)$. We choose parameters so that we can attain the lower bound in the case of repair by all available code symbols. Specifically, we set $k+t=129$ and choose $k=99$. As a result, in case of $n=256$ the scheme in Construction~\ref{construction:secret_sharing} attains the lower bound on private repair bandwidth. We also include in our graph the scheme in Construction~\ref{construction:secret_sharing} without privacy, which is essentially the original Dau-Milenkovic scheme~\cite{dau2021repairing} and a lower bound on the repair bandwidth without privacy requirements from the same paper. It can be observed that privacy constraints shift the graph of bandwidth vs. the number of available nodes to the right and slightly upward, indicating that we need to download more to perform private repair with the same number of involved nodes. We note that in plotting the graph hereinafter, we consider the minimum overall bandwidth for a smaller or equal number of available nodes and possible dimension of subspaces utilized to form subspace polynomials. 
 \begin{figure}[htb]
	\centering
	\includegraphics[width=0.5\textwidth]{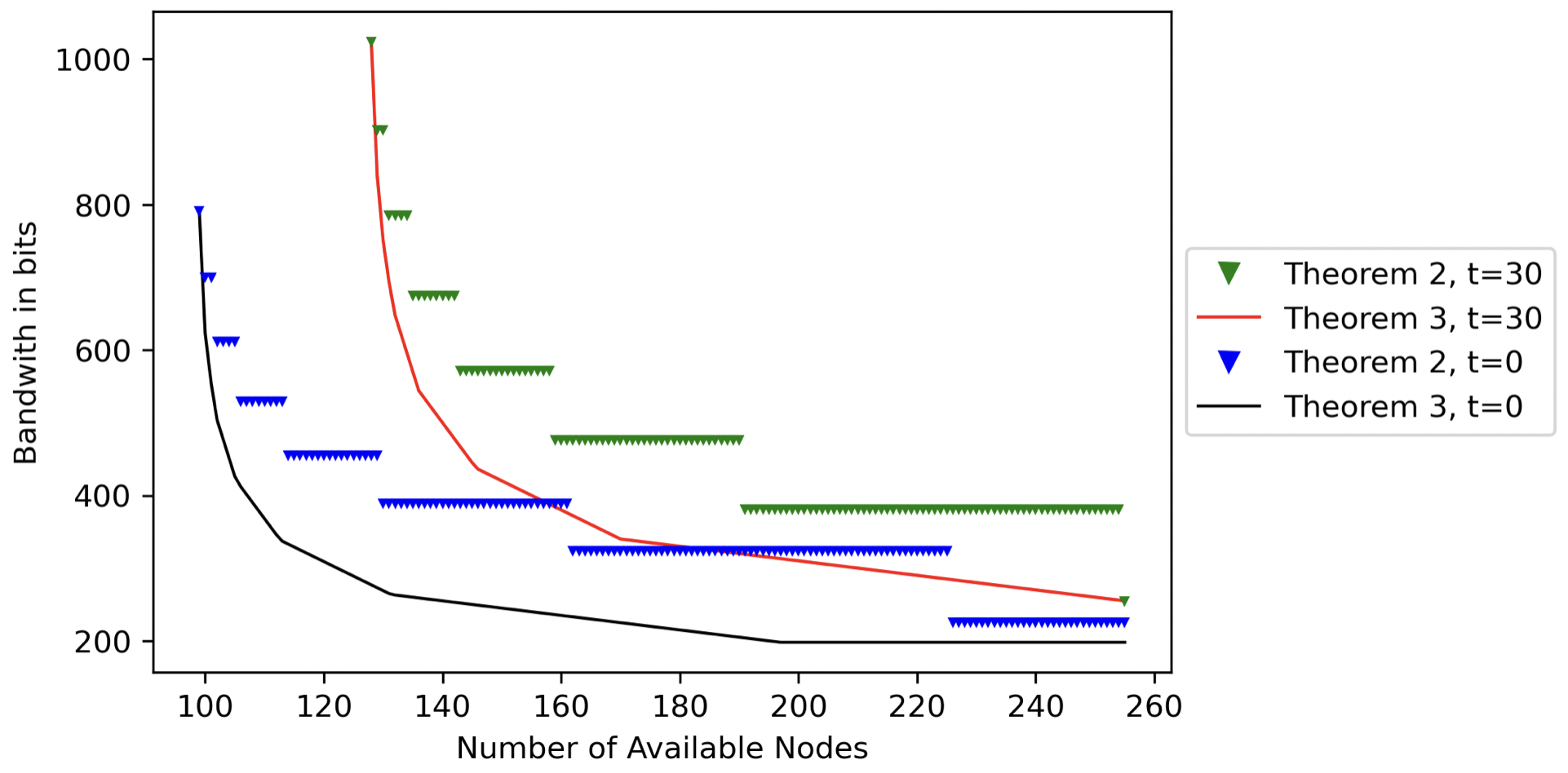}
	\caption{Recovery bandwidth for private repair of RS code with $k=99$ over $\textrm{GF}(2^{8})$}
	\label{2}
\end{figure}

\section{Conclusion}
\label{sec:conclusion}
We initiated the study of Reed-Solomon code repair with the additional goal of ensuring the privacy of the repaired index. We provided explicit schemes as well as a lower bound on total bandwidth, matching at some points. However, we view our work as just the tip of the iceberg, with many open problems and unanswered questions. For example, can we generalize our hidden subspace approach to several colluding nodes or ensure private repair with all nodes involved without an increase in bandwidth? Also, is it possible to perform recovery by a third party based on server responses only, or to generalize our approach for high sub-packetization regime?

{\section*{Acknowledgements.}
\textcolor{black}{The work of Stanislav Kruglik  was supported by the Ministry of Education, Singapore, under its MOE AcRF Tier~2 Award under Grant MOE-T2EP20121-0007. The work of Han Mao Kiah was supported by the Ministry of Education, Singapore, under its MOE AcRF Tier~2 Award under Grant MOE-T2EP20121-0007 and MOE AcRF Tier~1 Award under Grant RG19/23. The work of Son Hoang Dau is supported by the ARC DECRA Grant DE180100768.}
}

\bibliographystyle{IEEEtran}
\balance
\bibliography{ref_list}

\end{document}